\documentclass[a4paper,UKenglish]{lipics-v2018}
\usepackage{microtype}%if unwanted, comment out or use option "draft"
\usepackage[utf8]{inputenc}
\usepackage{amsfonts}
\usepackage{amsthm}
\usepackage{amssymb,amsmath}
\usepackage{wasysym}
\usepackage{graphics}
\usepackage{graphicx}
\usepackage{wrapfig, blindtext}
\usepackage{epstopdf}
\usepackage{algorithm} 
\usepackage{algorithmic}
\usepackage[x11 names, rgb]{xcolor}
\usepackage{tikz}
\usetikzlibrary{decorations,arrows,shapes,automata}
\usepackage{multicol}
\usepackage{hyperref}
\hypersetup{colorlinks=true}
\usepackage{color}

\title{The complexity of rational synthesis for concurrent games}

\author{Rodica Condurache}{Universit\'{e} Paris Est, LACL(EA 4219),
  UPEC\\{94010 Cr\'{e}teil Cedex,
    France}}{rodica.condurache@lacl.fr}{}{}
\author{Youssouf Oualhadj}{Universit\'{e} Paris Est, LACL(EA 4219),
  UPEC\\{94010 Cr\'{e}teil Cedex,
    France}}{youssouf.oualhadj@lacl.fr}{}{} 

\author{Nicolas Troquard}{The KRDB Research Centre, Free University of Bozen-Bolzano\\ {I-39100 Bozen-Bolzano BZ, Italy}}{nicolas.troquard@unibz.it}{}{}

\authorrunning{R. Condurache, Y. Oualhadj and N. Troquard}%mandatory. First: Use abbreviated first/middle names. Second (only in severe cases): Use first author plus 'et al.'

\Copyright{John Q. Public and Joan R. Public}%mandatory, please use full first names. LIPIcs license is "CC-BY";  http://creativecommons.org/licenses/by/3.0/

% mandatory: Please choose ACM 2012 classifications from https://www.acm.org/publications/class-2012 or https://dl.acm.org/ccs/ccs_flat.cfm . E.g., cite as "General and reference $\rightarrow$ General literature" or \ccsdesc[100]{General and reference~General literature}. 

%mandatory

\EventEditors{John Q. Open and Joan R. Access}
\EventNoEds{2}
\EventLongTitle{42nd Conference on Very Important Topics (CVIT 2016)}
\EventShortTitle{CVIT 2016}
\EventAcronym{CVIT}
\EventYear{2016}
\EventDate{December 24--27, 2016}
\EventLocation{Little Whinging, United Kingdom}
\EventLogo{}
\SeriesVolume{42}
\ArticleNo{23}
\nolinenumbers %uncomment to disable line numbering
\usepackage{todonotes}

\newcommand{\States}[0]{\mathsf{St}} 
 \newcommand{\Agt}[0]{\mathsf{Agt}}
\newcommand{\Act}[0]{\mathsf{Act}} \newcommand{\Tab}[0]{\mathsf{Tab}}
\newcommand{\obj}{\mathsf{Obj}} \newcommand{\Plays}[0]{\mathsf{Plays}}

\newcommand{\hist}[0]{\mathsf{Hist}}
\newcommand{\last}[0]{\mathsf{Last}}

\newcommand{\dev}{\mathsf{Deviator}}
\newcommand{\pref}{\mathsf{Root}}
\newcommand{\pay}{\mathsf{Payoff}}
\newcommand{\GtoH}{{\GGG2\HHH}}
\newcommand{\HtoG}{{\HHH2\GGG}}
\newcommand{\cons}{\mathsf{Eve}}
\newcommand{\spoil}{\mathsf{Adam}}
\newcommand{\astate}{{Q_\mathsf{A}}}
\newcommand{\bstate}{{Q_\mathsf{B}}}
\newcommand{\cstate}{{Q_\mathsf{C}}}
\newcommand{\dstate}{{Q_\mathsf{D}}}
\newcommand{\pr}{\mathsf{Pr}}

\def\GGG{\mathcal{G}} \def\HHH{\mathcal{H}} 

\def\out{{\sf Out}}
\def\set#1{\left\{#1\right\}}
\newtheorem{proposition}[theorem]{Proposition}
\newtheorem{problem}[theorem]{Problem}
\newtheorem{claim}[theorem]{Claim}
\begin{document}

\maketitle

\begin{abstract}
  In this paper, we investigate the rational synthesis problem for
  concurrent game structures for a variety of objectives ranging from
  reachability to Muller condition. We propose a new algorithm that
  establishes the decidability of the non cooperative rational
  synthesis problem that relies solely on game theoretic technique as
  opposed to previous approaches that are logic based.  Given an
  instance of the rational synthesis problem, we construct a zero-sum
  turn-based game that can be adapted to each one of the class of
  objectives. We obtain new complexity results. In particular, we show
  that in the cases of reachability, safety, Büchi, and co-Büchi
  objectives the problem is in \textsc{PSpace}, providing a tight
  upper-bound to the \textsc{PSpace}-hardness already established for
  turn-based games. In the case of Muller objective the problem is in
  \textsc{ExpTime}. We also obtain positive results when we assume a
  fixed number of agents, in which case the problem falls into
  \textsc{PTime} for reachability, safety, Büchi, and co-Büchi
  objectives.
\end{abstract}

\section{Introduction}

The synthesis problem aims at automatically designing a program from a
given specification. Several applications for this formal problem can
be found in the design of \emph{interactive systems} i.e., systems
interacting with an environment. From a formal point of view, the
synthesis problem is traditionally modelled as a zero-sum turn-based
game. The system and the environment are modeled by two players with
opposite interest. The goal of the system is the desired
specification.  Hence, a \emph{strategy} that allows the system to
achieve its goal against any behavior of the environment is a winning
strategy and is exactly the program to synthesize.

For a time, the described approach was the standard in the realm of
controller synthesis. However, due to the variety of systems to model,
such a pessimistic view is not always the most faithful one.  For
instance, consider a system that consists of a server and $n$ clients.
Assuming that all the agents have opposite interests is not a
realistic assumption. Indeed, from a design perspective, the purpose
of the server is to handle the incoming requests. On the other hand,
each client is only concerned with its own request and wants it
granted.  None of the agents involved in the described interaction
have antagonistic purposes. The setting of \emph{non-zero-sum games}
was proposed as model with more realistic assumptions.

In a non zero-sum game, each agent is equipped with a personal
objective and the system is just a regular agent in the game. The
agents interact together aiming at achieving the \emph{best
  outcome}. The best outcome in this setting is often formalized by
the concept of Nash equilibria. 
Unfortunately, a solution in this setting offers no guarantee
that a specification for a given agent is achieved, and in a synthesis
context one wants to enforce a specification for one a subset of the
agents.

The \emph{rational synthesis problem} was introduced as a
generalization of the synthesis problem to environment with multiple
agents~\cite{FismanKL10}. It aims at synthesizing a Nash equilibrium
such that the induced behavior satisfies a given specification.  This
vision enjoys nice algorithmic properties since it matches the
complexity bound of the classical synthesis problem.  Later on, yet
another version of the problem was proposed where the agents are
rational but not
cooperative~\cite{KuPV14,DBLP:journals/amai/KupfermanPV16}.  In the
former formalization, the specification is guaranteed as long as the
agents agree to behave according to the chosen equilibrium. But
anything can happen if not, in particular they can play another
equilibrium that does not satisfy the specification.  In the Non
Cooperative Rational Synthesis (NCRSP), the system has to ensure that
the specification holds in any equilibrium (c.f.,
Section~\ref{sec:ncrsp} for a formal definition and
Figure~\ref{fig:example} for an example).  A solution for both problems
was presented for specifications expressed in Linear Temporal Logic
(LTL).  The proposed solution relies on the fact that the problem can
be expressed in a decidable fragment of a logic called \emph{Strategy
  Logic}.  The presented algorithm runs in 2-\textsc{ExpTime}. While
expressing the problem in a decidable fragment of Strategy Logic gives
an immediate solution, it could also hide a great deal of structural
properties. Such properties could be exploited in a hope of designing
faster algorithms for less expressive objectives. In particular,
specifications such as reachability, liveness, fairness, \textit{etc.}

In~\cite{ICALP16}, the first author took part in a piece of work where
they considered this very problem for specific objectives such as
reachability, safety, Büchi, \textit{etc} in a turn-based interaction
model.  They established complexity bounds for each objective.

In this paper we consider the problem of non-cooperative rational
synthesis with concurrent interactions. We address this problem for a
variety of objectives and give exact complexity bounds relying
exclusively on techniques inspired by the theory of zero-sum games.
The concurrency between agents raises a formal challenge to overcome
as the techniques used in~\cite{ICALP16} do not directly extend. 
Intuitively, when
the interaction is turn-based, one can construct a tree
automaton that accepts solutions for the rational synthesis
problem. The nodes of an accepted tree are exactly
the vertices of the game. This helps a lot in dealing with deviations
but cannot be used in concurrent games.

In Section~\ref{sec:ncrsp}, we present an alternative algorithm
that solves the general problem for LTL specification. This algorithm
constructs a zero-sum turn-based game.  This fresh game is played
between \emph{Constructor} who tries to construct a solution and
\emph{Spoiler} who tries to falsify the proposed solution.
We then show in Section~\ref{sec:comp} how to use this algorithm to
solve the NCRSP for reachability, safety, Büchi, co-Büchi, and Muller
conditions. We also observe that we match
the complexity results for the NCRSP in turn-based games.

%%% Local Variables:
%%% mode: latex
%%% TeX-master: "main.tex"
%%% End:

\section{Preliminaries}
\subsection{Concurrent Game Structures}

\begin{definition}\label{def:cgs}
  A game structure is defined as a tuple
  $\GGG = (\States, s_0, \Agt, (\Act_i)_{i \in \Agt}, \Tab)$, where
  $\States$ is the set of states in the game, $s_0$ is the initial
  state, $\Agt = \{0,1,...,n\}$ is the set of agents, $\Act_i$ is the
  set of actions of Agent $i$,
  $\Tab : \States \times \prod_{i \in \Agt} \Act_i \rightarrow
  \States$ is the transition table.
\end{definition}

\begin{remark}
  Note that we consider game structures that are complete and
  deterministic. That is, from each state $s$ and any tuple of actions
  $ \bar{a} \in \prod_{i \in \Agt} \Act_i$, there is exactly one
  successor state $s'$.
\end{remark}

\begin{definition}
  A \emph{play} in the game structure is a sequence of states and
  actions profile $\rho = s_0 \bar{a}_0s_1 \bar{a}_1s_2\bar{a}_2 ...$
  in $(\States\ \prod_{i \in \Agt} \Act_i)^\omega$ such that $s_0$ is
  the initial state and for all $j \geq 0$,
  $s_{j+1} = \Tab(s_j, \bar{a}_j)$.
\end{definition}

Throughout the paper, for every word $w$, over any alphabet, we denote by $w[j]$ the $j+1$-th letter, and we denote by $w[0..j]$ the prefix of $w$ of size $j+1$.

By $\rho\restriction_{\States}$ we mean the projection of $\rho$ over
$\States$, and $\Plays(\GGG)$ is the set of all the plays in the game
structure $\GGG$.
We call history any finite sequence in
$\States\ (\prod_{i \in \Agt} \Act_i\ \States)^*$.  For a history $h$,
we denote by $h\restriction_{\States}$ its projection over $\States$, and
by $\last_{\States}(h)$ the last element of $h\restriction_{\States}$.
We denote by $\hist$ the set of all the histories.

In this paper we allow agents to see the actions played between
states. Therefore, they behave depending on the past sequence of
states and tuples of actions.  

\begin{definition}[Strategy and strategy profile]
A \emph{strategy} for Agent~$i$ is a
mapping
$
\sigma_i : \States \left(\prod_{i \in \Agt} \Act_i \ \States\right)^*
\rightarrow \Act_i
\enspace.
$
% Note that, since the transition's table is deterministic, we even
% forget the states along plays and write
% $
% \sigma_i : \States \left(\prod_{i \in \Agt} \Act_i \right)^* \rightarrow
% \Act_i
% \enspace.
% $

A \emph{strategy profile} is defined as a tuple of strategies
$\bar{\sigma} = \langle \sigma_0, \sigma_1,..., \sigma_n \rangle$ and 
by $\bar{\sigma}[i]$ we denote the strategy of $i$-th position (of
Agent $i$).

Also, $\bar{\sigma}_{-i}$ is the partial strategy profile obtained
from the strategy profile $\bar{\sigma}$
from which the strategy of Agent~$i$ is ignored. The tuple of strategies
$\langle\bar{\sigma}_{-i}, \sigma'_i\rangle$ is obtained from the
tuple $\bar{\sigma}$ by substituting Agent
$i$'s strategy with $\sigma'_i$. 
\end{definition}

Once a strategy profile is chosen it induces a play $\rho$.  We say
that a play $\rho = s_0 \bar{a}_0s_1 \bar{a}_1s_2\bar{a}_2 ... $ in
$(\States\ \prod_{i \in \Agt} \Act_i)^\omega$ is \emph{compatible}
with a strategy $\sigma_i$ of Agent~$i$ if for every prefix of
$\rho[0..2k]$ with $k \geq 0$, we have
$\sigma(\rho[0..2k]) = \bar{a}_{k}(i)$, where $\bar{a}_{k}(i)$ is the
action of Agent~$i$ in the vector $\bar{a}_{k}$.

%$\rho[0..n]$ of length $n\ge0$ we have $\sigma(\rho[0..n]) = \bar{a}_{n+1}(i)$ where $\bar{a}_{n+1}(i)$ is the action of Agent~$i$ in the vector $\bar{a}_{n+1}$.

We denote by $\Plays(\sigma_i)$ the set of all the plays that are
compatible with the strategy~$\sigma_i$ for
Agent~$i$. $\hist(\sigma_i)$ is the set of all the histories that are
compatible with $\sigma_i$. The outcome of
an interaction between agents following a certain strategy profile
$\bar{\sigma}$ defines a unique play in the game structure $\GGG$
denoted $\out(\bar{\sigma})$. It is the unique play in $\GGG$
compatible with all the strategies in the profile $\bar{\sigma}$
which is an infinite sequence over $(\States\ \prod_{i \in \Agt} \Act_i)$.  

\subsection{Payoff and Solution Concepts}
Each Agent $i \in \Agt$ has an objective expressed as a set $\obj_i$
of infinite sequences of states in $\GGG$.  
As defined before, a play $\rho$ is a sequence of states and action
profiles.  We slightly abuse notation and also write
$\rho \in \obj_i$, meaning that the sequence of states in the play
$\rho$ (that is, $\rho\restriction_\States$) is in $\obj_i$.  We
define the \emph{payoff} function that associates with each play
$\rho$ a vector $\pay(\rho) \in \{0,1 \}^{n+1}$ defined by
\[
  \forall i \in \Agt, \pay_i(\rho) = 1 \iff \rho\in \obj_i \enspace.
\]
We borrow game theoretic vocabulary and say that Agent $i$
\textit{wins} whenever her payoff is $1$. We sometimes abuse this
notation and write $\pay_i(\bar{\sigma})$, which is the payoff of
Agent $i$ associated with the \emph{unique} play induced by
$\bar{\sigma}$.

In this paper we are interested in winning objectives such as Safety,
Reachability, B\"{u}chi, coB\"{u}chi, and Muller that are defined as
follows.  Let $\rho$ be a play in a concurrent game structure
$\GGG$. We use the following notations:
\[
  occ(\rho) = \{ s \in \States \mid \exists j \geq 0 \text{ s.t. }
  \rho[j] = s \}
\]
to denote the set of states that appear along $\rho$ and
\[
  \inf(\rho) = \{ s \in \States \mid \forall j \geq 0, \exists k \geq
  j \text{ s.t. } \rho[k] = s \}
\]
to denote the set of states appearing infinitely often along
$\rho$. Then,

\begin{itemize}
\item \textit{Reachability:} For some $T \subseteq \States$,
  $ \textsc{Reach}(T) = \{ \rho \in \States^\omega \mid occ(\rho) \cap
  T \neq \emptyset \}$;

\item \textit{Safety:}For some $T \subseteq \States$,
  $ \textsc{Safe}(T) = \{ \rho \in \States^\omega \mid occ(\rho)
  \subseteq T \}$;

\item \textit{B\"{u}chi:} For some $T \subseteq \States$,
  $\textsc{B\"{u}chi}(T) = \{ \rho \in \States^\omega \mid \inf(\rho)
  \cap T \neq \emptyset \}$;

\item \textit{coB\"{u}chi:} For some $T \subseteq \States$,
  $ \textsc{coB\"{u}chi}(T) = \{ \rho \in \States^\omega \mid
  \inf(\rho) \cap T = \emptyset \}$;

\item \textit{Parity:} For some priority function
  $p : \States \rightarrow \mathbb{N}$,
  $ \textsc{Parity}(p) = \{ \rho \in \States^\omega \mid \min\{ p(s)
  \mid s \in \inf(\rho)\} \text{ is even } \}$;

\item \textit{Muller:} For some boolean formula $\mu$ over $\States$,
  $ \textsc{Muller}(\mu) = \{ \rho \in \States^\omega \mid \inf(\rho)
  \models \mu \}$.
\end{itemize}
A Nash equilibrium is the formalisation of a situation where no agent
can improve her payoff by unilaterally changing her
behaviour. Formally:

\begin{definition}(Nash equilibrium) A strategy profile $\bar{\sigma}$
  is a Nash equilibrium (NE) if for every agent $i$ and every strategy
  $\sigma'$ of $i$ the following holds true:
  \[
    \pay_i(\bar{\sigma}) \ge
    \pay_i(\langle\bar{\sigma}_{-i},\sigma'\rangle) \enspace.
  \]
\end{definition}

Throughout this paper, we will assume that Agent~$0$ is the agent for
whom we wish to synthesize the strategy, therefore, we use the concept
of 0-fixed Nash equilibria.

\begin{definition}[0-fixed Nash equilibrium]
A profile $\langle\sigma_0,\bar{\sigma}_{-0}\rangle$ is a 0-fixed NE (0-NE),
if for every strategy $\sigma'$ for agent $i$ in $\Agt\setminus
\set{0}$ the following holds true:
\[
\pay_i(\langle\sigma_0,\bar{\sigma}_{-0}\rangle) \ge
\pay_i(\langle \sigma_{0}, {(\bar{\sigma}_{-0})}_{-i},\sigma'\rangle)
\enspace.
\]
\end{definition}
That is, fixing $\sigma_0$ for Agent~$0$, the other agents cannot
improve their payoff by unilaterally changing their strategy.

\subsection{Rational synthesis}

The rational synthesis can be defined in a optimistic or pessimistic
setting. The former one is the so-called Cooperative Rational
Synthesis (CRSP) Formally defined  as
\begin{problem}
  \label{pb:NCRSP}
  Is there a 0-NE $\bar{\sigma}$ 
  such that $\pay_0(\bar{\sigma}) = 1$?
\end{problem}

The latter is the so-called  Non Cooperative Rational Synthesis
Problem (NCRSP) and is formally defined as

\begin{problem}
  \label{pb:NCRSP}
  Is there a strategy $\sigma_0$ for
  Agent $0$ such that for every $0$-NE
  $\bar{\sigma} = \langle\sigma_0, \bar{\sigma}_{-0} \rangle $,
  we have $\pay_0(\bar{\sigma}) = 1$?
\end{problem}

In this paper we study computational complexity for the rational
synthesis problem in both cooperative and non-cooperative settings.

For the CRSP, the complexity results are corollaries of existing
work. In particular, for Safety, Reachability, B\"{u}chi,
co-B\"{u}chi, Rabin and Muller objectives, we can apply algorithms
from~\cite{BBMU15} to obtain the same complexities for CRSP as for the
turn-based models when the number of agents is not fixed.  More
precisely, in~\cite{BBMU15} the problem of finding NE in concurrent
games is tackled. In this problem one asks for the existence of NE
whose payoff is between two thresholds. Then, by choosing the lower
thresholds to be such that only Agent 0 satisfies her objective and
the upper thresholds such that all agents win, we reduce to the
cooperative rational synthesis problem.  Brenguier et
al.~\cite{BBMU15} showed that the existence of constrained NE in
concurrent games can be solved in \textsc{PTime} for B\"{u}chi
objectives, NP for Safety, Reachability and coB\"{u}chi objectives,
and \textsc{PSpace} for Muller objectives.  All hardness results are
inferred directly from the hardness results in the turn-based
setting. This is a consequence of the fact that every turn-based game
can be encoded as a concurrent game by allowing at each state at most
one agent to have non-vacuous choices. For Streett objectives, by
reducing to~\cite{BBMU15} we only obtain \textsc{PSpace}-easiness and
the $NP$-hardness comes from the turn-based setting~\cite{ICALP16}.

In the case of non-cooperative rational synthesis, we cannot directly
apply the existing results. However, we define an algorithm
inspired from the \emph{suspect games}~\cite{BBMU15}.  The suspect game was
introduced to decide the existence of pure NE in
concurrent games with $\omega$-regular objectives. We inspire
ourselves from that approach and design a zero-sum game that
combines the behaviors of Agent 0 and an extra entity whose goal is to
prove, when needed, that the current play is not the outcome of a
0-NE.  We also extend the idea in~\cite{ICALP16} that consists roughly in
keeping track of deviations. Recall that the non-cooperative rational
synthesis problem consists in designing a strategy $\sigma_0$ for the
protagonist (Agent~0 in our case) such that her objective $\obj_0$
is satisfied by all the plays that are outcomes of 0-NE compatible
with $\sigma_0$.  This is equivalent to finding 
a strategy $\sigma_0$ for Agent~0 such that for any play $\rho$
compatible with it, either $\rho$ satisfies $\obj_0$, or there is no
strategy profile
$\bar{\sigma} = \langle\sigma_0, \bar{\sigma}_{-0} \rangle$ that is a
0-NE whose outcome is $\rho$.

 \begin{figure}[h!]
   \begin{subfigure}[b]{0.4\textwidth}
    \begin{tikzpicture}[>=latex', join=bevel, thick, initial text =,
      scale = .7]
      %%% states
      \node[state] (0) at (0bp, 50bp) [draw, circle, initial above]{$s_0$}; 
      \node[state] (1) at (75bp, 0bp) [draw,circle] {$s_1$};
      \node[state] (2) at (-75bp, 0bp) [draw, circle] {$s_2$};
      \node[state] (3) at (75bp, -75bp) [draw, rectangle, rounded corners] {$T_{\{0,1\}}$};
      \node[state] (4) at (-75bp, -75bp) [draw, rectangle, rounded corners] {$T_{\{2\}}$};
      %%% edges
      \draw[->] (0) [bend right] to node [left] {
        \begin{tabular}{c}
          $(l,*,*)$
        \end{tabular}} (2); 
      \draw[->] (0) to node [above right] {
        \begin{tabular}{c}
          $(r,*,*)$
        \end{tabular}} (1); 
      \draw[->] (2) to node [left] {
        \begin{tabular}{c}
          $(l,*,b)$
        \end{tabular}} (4); 
      \draw[->] (2) to node [right] {
        \begin{tabular}{c}
          $(l,*,a)$
        \end{tabular}} (0); 
      \draw[->] (1) to node [below] {
        \begin{tabular}{c}
          $(r,b,*)$
        \end{tabular}} (2); 
      \draw[->] (1) to node [right] {
        \begin{tabular}{c}
          $(r,a,*)$
        \end{tabular}} (3); 
      \draw[->] (1) edge [loop right] node [right] {$(l,*,*)$} (1);
      \draw[->] (2) edge [loop left] node [left] {$(r,*,*)$} (2);
    \end{tikzpicture}
    \caption{\label{fig:example} A concurrent game.}
  \end{subfigure}
  \hfill
  \begin{subfigure}[b]{0.45\textwidth}
    \begin{tikzpicture}[>=latex', join=bevel, thick, initial text =,
      scale = .7]
      %%% states
      \node[state] (0) at (0bp, 50bp) [draw, circle, initial above]{$s_0$}; 
      \node[state] (1) at (75bp, 0bp) [draw,circle] {$s_1$};
      \node[state] (2) at (-75bp, 0bp) [draw, circle] {$s_2$};
      \node[state] (3) at (75bp, -75bp) [draw, rectangle, rounded corners] {$T_{\{0,1\}}$};
      \node[state] (4) at (-75bp, -75bp) [draw, rectangle, rounded corners] {$T_{\{2\}}$};
      %%% edges; 
      \draw[->] (0) to node [above right] {
        \begin{tabular}{c}
          $(r,*,*)$
        \end{tabular}} (1); 
      \draw[->, dotted] (2) to node [left] {
        \begin{tabular}{c}
          $(l,*,b)$
        \end{tabular}} (4); 
      \draw[->, dotted] (2) to node [right] {
        \begin{tabular}{c}
          $(l,*,a)$
        \end{tabular}} (0); 
      \draw[->, dashed] (1) to node [below] {
        \begin{tabular}{c}
          $(r,b,*)$
        \end{tabular}} (2); 
      \draw[->, dashed] (1) to node [right] {
        \begin{tabular}{c}
          $(r,a,*)$
        \end{tabular}} (3); 
    \end{tikzpicture}
    \caption{\label{fig:subGame} Subgame induced from the strategy
      $\sigma_0$.}
  \end{subfigure}
\end{figure}

\begin{example}
  Consider the concurrent game with reachability objectives depicted
  in Figure~\ref{fig:example}. The game starts in the state $s_0$.
  There are three agents, the controller Agent~$0$, Agent~$1$, and
  Agent~$2$.  Agent 0 has two actions $r$ for right and $l$ for left.
  Agents 1 and 2 have two actions, denoted $a$ and $b$.  For any
  subset $C$ of $\{0,1,2\}$, the states $T_{C}$ indicate that the
  agents in $C$ have reached their objectives (These states are
  sinks). In addition, there are three states $s_0$, $s_1$, and $s_2$.
  The edges represent the transitions table. The labels indicate the
  action profiles e.g.  the vector $(r,a,b)$ means that Agent 0 took
  action $r$, Agent 1 took action $a$, and Agent 2 took action
  $b$. Finally action $*$ stands for the indifferent choice that is
  any action for a given agent.  We can see that at $s_0$, Agent~$0$
  is the only agent with non-vacuous choices. He can choose to go to
  $s_1$ by playing action~$r$, or to go to $s_2$ by playing
  action~$l$.

  Now consider the strategy $\sigma_0$ for Agent~$0$ defined as
  follows: $\sigma_0(s_0) = r, \sigma_0(s_1) = r, \sigma_0(s_2) = l$
  We argue that this strategy is a solution to the NCRSP.  Indeed, by
  applying this strategy, we obtain the subgame of
  Figure~\ref{fig:subGame}. In this game, all the plays falsifying the
  objective of Agent 0 are the ones where Agent 1 plays $b$. Notice
  now that these plays are not outcomes of a 0-NE since Agent 1 can
  deviate by playing action $a$.
\end{example}
  
%%% Local Variables:
%%% mode: latex
%%% TeX-master: "main.tex"
%%% End:

%%% Local Variables:
%%% mode: latex
%%% TeX-master: "main.tex"
%%% End:

\section{Solution for Problem~\ref{pb:NCRSP}}
\label{sec:ncrsp}
We will now describe a general algorithm that solves the NCRSP. As a 
first step in our procedure, we construct a two-player turn based
game.

\subsection{Construction of a two-player game}
Given a concurrent game $\GGG = (\States, s_0, \Agt, (\Act_i)_{i \in \Agt}, \Tab)$ we construct a turn-based 2-player zero-sum game
$\HHH = ( Q, q_0, \Act_E, \Act_A, \Tab', \obj) $.

The game $\HHH$ is obtained as follows:
\begin{itemize}
\item $q_0 = (s_0, \emptyset, \emptyset)$
\item The set $\Act_E$ is $\Act^a_E \cup \Act^c_E$ where:
  \begin{itemize}
  \item $\Act^a_E = \Act_0 \times \prod_{i =1}^n (\Act_i \cup \{-\})$
  \item $\Act^c_E = \prod_{i =1}^n (\Act_i \cup \{-\})$.
  \end{itemize}

\item The set $\Act_A$ is $\prod_{i =1}^n \Act_i$.

\item The set $Q$ of states is $\astate \cup \bstate \cup \cstate \cup \dstate$ where
  \begin{align*}
    \astate & =  \States \times 2^\Agt \times 2^\Agt\\
    \bstate &=  \States \times 2^\Agt \times 2^\Agt \times \Act^a_E\\
    \cstate & =  \States \times 2^\Agt \times 2^\Agt \times \Act^a_E \times \Act_A\\
    \dstate & =  \States \times 2^\Agt \times 2^\Agt \times \Act^a_E \times \Act_A \times \Act^c_E \enspace.
  \end{align*}
\item Player $\cons$ plays in the states in  $\astate$ and $\cstate$, while Player $\spoil$ plays in the states in $\bstate$ and $\dstate$.
The legal moves are given as follows:
\begin{itemize}
\item From a state
  $ (s,W,D) \in \astate$,
  $\cons$ plays an action
  \[
  \bar{a} \in \Act_0 \times \prod_{i =1}^n (\Act_i \cup \{-\})
  \text{ s.t. } \forall 1 \leq i \leq n, \ %\bar{a}[i] \in \Act_i \Leftrightarrow i \in W
  \bar{a}[i] = -  \Longleftrightarrow i \not\in W
  \enspace.
  \]    
\item From a state $(s,W,D,\bar{a}) \in \bstate$, $\spoil$ plays an action $\bar{b} \in \Act_A$.
\item From a state
  $(s,W,D,\bar{a},\bar{b}) \in \cstate$, $\cons$ plays an action
  \[
  \bar{c} \in \prod_{i =1}^n (\Act_i \cup \{-\}) \text{ s.t. }
  %\bar{c}[i] \in \Act_i \implies (i \not\in W\cup D)
  i \in W\cup D \implies \bar{c}[i] = -
  \enspace.
  \]
\item From a state $(s,W,D,\bar{a}, \bar{b}, \bar{c}) \in \dstate$, $\spoil$ plays an action $\bar{d} \in \Act_A$. 
\end{itemize}
\end{itemize}

The transition $\Tab'$ and the objective $\obj$ of the game $\HHH$ are
described next.

\subsection{Transition function}
The game $\HHH$ is best understood as a dialogue between $\cons$ and
$\spoil$.
In each state
$(s,W,D)$ $\cons$ proposes an action for Agent 0 together
with the actions corresponding to the winning strategies of the agents
in the set $W$. Then, $\spoil$ responds with an action profile
played by all agents in the environment.  In the next step,
$\cons$ knows the entire action profile played by the
agents and proposes some new deviations for the agents that do not
have a deviation yet (they are neither in $W$ nor in $D$).  The last
move is performed by $\spoil$, it is his role to ``check''
that the proposed deviations and winning strategies are
correct. Therefore, $\spoil$ can choose any continuation for
the game and the sets $W$ and $D$ are updated according to the
previous choices to some new values $W'$ and $D'$.
 Each dialogue ``round'' is decomposed into four moves.

%% From a state $(s,W,D) \in \astate$
%% From a state $(s,W,D,\bar{a}) \in \bstate$
%% From a state $(s,W,D,\bar{a},\bar{b}) \in \cstate$
%% From a state $(s,W,D,\bar{a},\bar{b},\bar{c}) \in \dstate$

The transitions are given by the (partial) function $\Tab': Q \times (\Act_E \cup \Act_A) \to Q$:
  \begin{itemize}
  \item When $(s,W,D) \in \astate$, $\Tab'((s,W,D), \bar{a}) = (s,W,D,\bar{a})$.
  \item When $(s,W,D,\bar{a}) \in \bstate$, $\Tab'((s,W,D,\bar{a}), \bar{b}) = (s,W,D,\bar{a},\bar{b})$.
  \item When $(s,W,D,\bar{a},\bar{b}) \in \cstate$, $\Tab'((s,W,D,\bar{a},\bar{b}), \bar{c}) = (s,W,D,\bar{a},\bar{b}, \bar{c})$.
  \item When $(s,W,D,\bar{a},\bar{b},\bar{c}) \in \dstate$, $\Tab'((s,W,D,\bar{a},\bar{b}, \bar{c}), \bar{d}) = (s',W',D')$, such that:
    \begin{itemize}
      \item $s' = \Tab(s,(\bar{a}[0],\bar{d}))$. 
      \item
        $W' = W \cup \left\{i \not\in W\cup D \mid (\bar{d}[i] =
          \bar{c}[i]) \text{ and }
          (\forall j \in \Agt \setminus \{0,i\}, \bar{d}[j] = \bar{b}[j])\right\}\\
        \setminus \left\{i \in W \mid \bar{d}[i] \not =
          \bar{a}[i]\right\}$. That is, Agent $i$ is added to the set
        $W'$ on the continuations where Agent $i$ plays the new action
        proposed by $\cons$ in $\bar{a}$ (supposedly compatible with a
        winning strategy) and the other agents do not change their
        actions with respect to $\bar{d}$. Also, any agent for whom
        $\cons$ proposes an action in $\bar{c}$ is a hint to $\spoil$
        that this agent can deviate from that point. It is up to
        $\spoil$ to agree or not. If $\spoil$ agrees, we say that he
        has agreed with the recommendation of $\cons$. In this case,
        $\cons$ has to prove that she made the right choice, this will
        be checked by the winning condition of the game.
      \item
        $D' = D \cup \{i \in W \mid \bar{d}[i] \not = \bar{a}[i]\} \\
        \cup \left\{i \not \in W \cup D \mid (\bar{c}[i] \neq -)
          \text{ and } (\bar{d}[i] \not = \bar{c}[i]) \text{ and }
          (\forall j \in \Agt\setminus \{0,i\}, \bar{d}[j] =
          \bar{b}[j]) \} \right\}$. This is the opposite case where
        $\spoil$ stood by his choices, in this case the winning
        condition has to check that this was a wrong decision.
    \end{itemize}
  \end{itemize}

\subsection{Winning condition}

We equip $Q$ with the canonical projection $\pi_i$ that is the
projection over the $i$-th component. In particular, for every
$(s,W,D) \in \astate$, we have $\pi_1((s,W,D)) = s$,
$\pi_2((s,W,D)) = W$, and $\pi_3((s,W,D)) = W$.  We also extend
$\pi_i$ over $Q^+$ and $Q^\omega$ as expected.  Histories for $\cons$
are finite words in $q_0(\Act_EQ\Act_AQ)^*$.  Histories for $\spoil$
are finite words in $q_0(\Act_EQ)^*$.  Plays are infinite sequence in
$q_0(\Act_EQ\Act_AQ)^\omega$. Let $r$ be a play, we denote
$r\restriction_\astate$ the restriction of $r$ over the states in
$\astate$ which is an infinite sequence in $\astate^\omega$.  The set
$\lim \pi_2(r\restriction_\astate)$ (resp.\
$\lim \pi_3(r\restriction_\astate)$) is the set of agents in the limit
of $W$'s (resp.\ $D$'s). The limit $\lim \pi_3(r\restriction_\astate)$
exists because the sets $D$ occurring in the states $Q$ along a play
are non-decreasing subsets of $\Agt$, and $\Agt$ is finite. The limit
$\lim \pi_2(r\restriction_\astate)$ exists because (1)~an agent is
added into $W$ only if it is not in $D$, and (2)~when an agent leaves
$W$, it gets into $D$ indefinitely. This means that when an agent
leaves from $W$, it never goes back.

We define the following sets:
\begin{align}
  \label{eq:1}
  &S_0=\set{r\in Q^\omega \mid \pi_1(r\restriction_\astate) \in
    \obj_0} \enspace,\\ 
  \label{eq:2}
  &S_W= \set{r\in Q^\omega \mid\forall i\in \lim\pi_2(r\restriction_\astate), \pi_1(r\restriction \astate) \in
    \obj_i}\enspace,\\ 
  \label{eq:3}
  &S_D=\set{r\in Q^\omega \mid\exists i\in \lim\pi_3(r\restriction_\astate), \pi_1(r\restriction \astate)
    \not\in \obj_i}\enspace. 
\end{align}
\[
  \obj = (S_0 \cup S_D)\cap S_W
  \enspace.
\]

\subsection{Transformations}
\paragraph*{Lifting of histories}
\label{sec:GtoH}
We define a transformation over histories in $\GGG$ to create histories
in $\HHH$.
For every strategy $\sigma$ for $\cons$ in $\HHH$, we define the
transformation $\GtoH_\sigma$.
% noted simply $\GtoH$ when there is no ambiguity. 

Let $h$ be a history in $\GGG$ and assume that
$h = s_0\bar{m}_0s_1\bar{m}_1 ... s_k\bar{m}_{k}s_{k+1}$.
The lifting of $h$ is a history $\tilde{h}$ in $\HHH$ obtained by
the mapping $\GtoH_\sigma$ inductively defined as
follows:
\[
  \GtoH_\sigma(s_0) = (s_0,\emptyset,\emptyset)
  \enspace,
\]
and 
% \[
%   \GtoH_\sigma(s_0\bar{m}_0s_1\bar{m}_1 ... s_{k}) = \tilde{h}
%   \enspace,
% \]
\[
  \GtoH_\sigma(h) = 
  \underbrace{\GtoH_\sigma(s_0\bar{m}_0s_1\bar{m}_1 ... s_k)}_{\tilde{h'}}
  \bar{a}q_b \bar{b}q_c\bar{c}q_d \bar{d}q_a
  \enspace,
\]
% where
% \[
%   \GtoH_\sigma(h) = \tilde{h}\bar{a}q_b
%   \bar{b}q_c \bar{c}q_d \bar{d}q_a 
% \]
where 
\begin{align*}
  &\bar{a} = \sigma(\tilde{h'}) \enspace,&
  &q_b  = \Tab'(\last(\tilde{h'}),\bar{a}) \enspace,\\
  &\bar{b} = \bar{m}_{k_{-0}} \enspace,&
  &q_c  = \Tab'(q_b,\bar{b}) \enspace,\\
  &\bar{c} =  \sigma(\tilde{h'}\bar{a}q_b \bar{b}q_c) \enspace,&
  &q_d = \Tab'(q_c,\bar{c}) \enspace,\\
  &\bar{d} = \bar{b} = \bar{m}_{k_{-0}} \enspace,&
  &q_a = \Tab'(q_d,\bar{d})\enspace.
\end{align*}

Observe that every history $\GtoH_\sigma(h)$ ends in a state in $\astate$, where $\cons$ plays an action from $\Act^a_E$, that always specifies an action for Agent~$0$.
The function $\GtoH_\sigma$ is thus instrumental in obtaining a strategy $\sigma_0$ for Agent~$0$ in $\GGG$ from a strategy of Player $\cons$ in $\HHH$. For every history $h$ in $\GGG$, we define:
\begin{align}
  \label{eq:sigma0}
  \sigma_0(h) = \sigma(\GtoH_\sigma(h))[0]
  \enspace.
\end{align}

For every strategy $\sigma$ of $\cons$, we call 0-strategy the
strategy obtained by Equation~\ref{eq:sigma0}.
The following claim is consequence of the same equation.

\begin{claim}
  \label{claim:gtoh-hist}
  Let $\sigma$ be a strategy for $\cons$, and let $\sigma_0$ be the
  $0$-strategy. 
  If a history $h$ in $\GGG$ is compatible with $\sigma_0$ then
  the history $\tilde{h} = \GtoH_\sigma(h)$ in $\HHH$ is compatible
  with $\sigma$.
\end{claim}

The function $\GtoH_\sigma$ maps every history in $\GGG$ into a history in $\HHH$. We define $\GtoH^\bullet_\sigma$ as the natural extension of $\GtoH_{\sigma}$ over the domain of plays in $\GGG$.
We extend the previous claim as expected.

\begin{claim}
  \label{claim:gtoh-play}
  Let $\sigma$ be a strategy for $\cons$, and let $\sigma_0$ be the
  $0$-strategy.
  If a run $\rho$ in $\GGG$ is compatible with $\sigma_0$ then
  the run $r = \GtoH^\bullet_\sigma(\rho)$ in $\HHH$ is compatible
  with $\sigma$.
\end{claim}

\begin{lemma} 
  \label{lm:GtoH-proj}
  Let $\sigma$ be a strategy for $\cons$, let $\rho$ be a run in $\GGG$ compatible with the 0-strategy
  $\sigma_0$. Let $h$ be a history in
  $\GGG$, assume $h$ to be a prefix 
  of $\rho$. If $\tilde{h} = \GtoH_{\sigma}(h)$ then 
  $\pi_1(\tilde{h}\restriction_\astate) = h\restriction_\States$.
\end{lemma}

\begin{proof}
  By induction on the size of $h$. The base case is $h = s_0$, in
  which case
  $ \GtoH_{\sigma}(h) = \tilde{h} = (s_0,\emptyset,\emptyset) $.  We
  have
  $ \pi_1(\tilde{h} \restriction_\astate) = s_0 = h
  \restriction_\States $.  Now assume for induction that
  $ \pi_1(\tilde{h} \restriction_\astate) = h \restriction_\States $
  for every history $h = s_0\bar{m}_0s_1\bar{m}_1 ... s_{k}$ of size
  $1+2k$ and let $\GtoH_{\sigma}(h) = \tilde{h}$.

  Now consider the history $hm_ks_{k+1}$ by definition
  $ \GtoH_{\sigma}(hm_ks_{k+1}) = \tilde{h}\bar{a}q_b\bar{b}q_c
  \bar{c}q_d \bar{d}q_a $ where $\bar{a}, \bar{b}, \bar{c}, \bar{d}$
  are obtained thanks to $\GtoH_{\sigma}$, by I.H.
  $ \pi_1(\tilde{h}) = s_0s_1 ... s_k$, it thus suffices to show that
  $ \pi_1(q_a) = s_{k+1} $. For this, one needs to 
  remark that $m_k = (\sigma(\tilde{h})[0], \bar{d})$, and that
   \[
    s_{k+1} = \Tab(s_k, m_k) =
    \pi_1(\Tab'((s,W,D,\bar{a},\bar{b}, \bar{c}), \bar{d})) =
    \pi_1(q_a)
  \]
  where the second equality is by definition of the construction.
\end{proof}

Since the previous lemma is true for any histories that are
respectively prefixes of $r$ and $\rho$ we obtain the following claim:

\begin{claim}
  \label{claim:GtoH-proj-runs}
  Let $\sigma$ be a strategy for $\cons$, let $\rho$ be a run in
  $\GGG$ compatible with the 0-strategy 
  $\sigma_0$. If $r = \GtoH_{\sigma}^\bullet(\rho)$ then 
  $\pi_1(r\restriction_\astate) = \rho\restriction_\States$.
\end{claim}

\paragraph*{Projection of histories}
We now define in some sense the reverse operation.
Let us define the transformation $\HtoG$.

Let $\tilde{h}$ be a history in $\HHH$ ending in a state in $\astate$.
\[
\HtoG(q_0) = s_0
\]
\[
  \HtoG(\tilde{h}\bar{a}q_b \bar{b}q_c\bar{c}q_d\bar{d}q_a)
  = \underbrace{\HtoG(\tilde{h})}_{\text{\tiny{induction}}}
  \underbrace{(\bar{a}[0], \bar{d}_{-0})}_{\text{\tiny{action}}}
  q_a
  % \underbrace{\Tab(\last(\HtoG(\tilde{h})) ,(\bar{a}[0],
  %   \bar{d}_{-0}))}_{\text{\tiny{next state}}}
\]

\begin{lemma}
  Let $\tilde{h}$ be a run in $\HHH$, $h$ be a history in $\GGG$. If
  $h = \HtoG(\tilde{h})$, then
  $\pi_1(\tilde{h} \restriction_\astate) = h\restriction_\States$
\end{lemma}

\begin{proof}
  By induction over the length of $\tilde{h}$. For
  $\tilde{h} = (s_0,\emptyset,\emptyset)$ the result 
  trivially true.  
  Assume the result holds for any history $\tilde{h}$ and let us show that
  it holds for $\tilde{h}\bar{a}q_b \bar{b}q_c\bar{c}q_d\bar{d}q_a$. By
  induction we have $\pi_1(\tilde{h} \restriction_\astate) =
  h\restriction_\States$, to conclude notice that
  \[
    \pi_1(q_a) = 
    \Tab(\last(\HtoG(\tilde{h})) ,(\bar{a}[0],\bar{d}_{-0}))
  \]
  % it follows that 
  % \[
  %   \pi_1(\tilde{h}\bar{a}q_b \bar{b}q_c\bar{c}q_d\bar{d}q_a
  %   \restriction_\astate)
  %   =
  %   \left(
  %     h\cdot (\bar{a}[0],\bar{d}_{-0}) \cdot 
  %     \Tab\left(\last(h)
  %   ,(\bar{a}[0],\bar{d}_{-0})\right)
  %   \right) \restriction_\States 
  % \]
\end{proof}

The function $\HtoG$ maps every history in $\HHH$ ending in a state in
$\astate$ into a history in $\GGG$. We define $\HtoG^\bullet$ as the
natural extension of $\HtoG$ over the domain of runs in $\HHH$.

The following claim follows

\begin{claim}
  \label{claim:HtoG}
  Let $r$ be a run in $\HHH$, $\rho$ be a run in $\GGG$. If $\rho =
  \HtoG^\bullet(r)$, then
  $\pi_1(r \restriction_\astate) = \rho\restriction_\States$
\end{claim}

%The strategy $\tilde{\sigma}$ of $\cons$ is defined as follows:

% For the energy objective we define the set

% \[
%   E\obj = \obj \cap \energy(c_0)
%   \enspace,
% \]
% where $c_0$ is an initial credit.

%%% Local Variables:
%%% mode: latex
%%% TeX-master: "main.tex"
%%% End:

\section{Main Theorem}

\begin{theorem}
  There exists a solution for the NCRSP iff $\cons$ wins.
\end{theorem}

We denote $\sigma_i^h$ the strategy that mimics the strategy
$\sigma_i$ when the current history is~$h$ i.e.
\begin{align*}
  \sigma_i^{h}(h') = 
  \begin{cases}
    \sigma_i(h') \text{ if $h'$ is a prefix of $h$}\\
    \sigma_i(h\cdot h') \text{ if $h$ is a prefix of $h'$}\\
    \bot \text{ otherwise}
  \end{cases}
\end{align*}

\begin{definition}
  \label{def:devpoint}
  Let $\rho$ be a play and let 
  $h = s_0 \bar{a}^0 s_1 \bar{a}^1 \cdots s_k
  \bar{a}^k s_{k+1}$
  be a prefix of $\rho$.  We say that $h$
  is a \textit{good deviation point}
  for Agent $i \in \Agt \setminus \{0\}$ if:
  \begin{itemize}
  \item $\rho\restriction_\States \not \in \obj_i$ and,
  \item there exists a strategy $\sigma'_i$ of Agent $i$ from $[h]$
    such that for all $(\sigma_j)_{j \in \Agt \setminus \{0,i\}}$ we
    have:
    \[
    [h] \cdot \out\left(\sigma_0^{[h]},...,\sigma'_i,...,
      \sigma_n^{[h]}\right) \in \obj_i \enspace, \text{ where }
    \]
    \[
    [h] = \rho[0..k]\cdot \langle \bar{a}^k_{-i},
    \sigma'_i(\rho[0..k]) \rangle \cdot \Tab\left(s_k,\langle
      \bar{a}^k_{-i}, \sigma'_i(\rho[0..k]) \rangle \right)
    \enspace.
    \]
  \end{itemize}
  We say that $\rho$ has a \textit{good deviation} if some prefix $h$
  of $\rho$ is a good deviation point.
\end{definition}

We use the notion of deviation point in the following lemma.
This lemma states that a strategy $\sigma_0$ is a solution
for the NCRSP if any play $\rho$ compatible with it, either is winning
for Agent 0 or some  Agent $i$ would unilaterally deviate
and win against any strategy profile of the other agents.

\begin{lemma} \label{lemma:devPoint}
  A strategy $\sigma_0$ is a solution for NCRSP iff every play $\rho$
  compatible with $\sigma_0$ either
  $\rho \restriction_{\States} \in \obj_0$ or, $\rho$ has a good
  deviation.
\end{lemma}

\begin{proof}
   We start by establishing the if direction, let
    $\sigma_0$ be a solution for the NCRSP.  If any outcome $\rho \in \Plays(\sigma_0)$ is such that  $\rho \restriction_{\States} \in \obj_0$ then there is nothing to prove.  
    Let $\rho$ be a play in $\Plays(\sigma_0)$ such that $ \rho$
    is not in $\obj_0$. Assume toward a contradiction that $\rho$
    does not contain a good deviation point. Then by
    Definition~\ref{def:devpoint} we know that for any prefix $h$ of
    $\rho$, any agent $i\neq 0$ such that $\pay_i(\rho) = 0$, and any
    strategy $\tau_i$ of $i$ there exists $\sigma_1,\cdots,\sigma_n$
    strategies for agents 1 to $n$ such that the following holds:
    \[
    [h] \cdot \out\left(\sigma_0^{h},
      \sigma_1^{h},\cdots,\tau_i^{h},\cdots, \sigma_n^{h}\right)
    \not\in \obj_i \enspace.
    \]
    The above equation implies that Agent $i$ does not have a
    profitable deviation under the strategy $\sigma_0$, hence the
    profile $\langle\sigma_0,\cdots,\sigma_n\rangle$ is a 0-fixed NE
    contradicting the fact that $\sigma_0$ is a solution for the
    NCRSP. 
     
     For the only if direction, let $\sigma_0$ be a
    strategy for agent 0, assume that every $\rho$ in
    $\Plays(\sigma_0)$ satisfies
    \begin{enumerate}
    \item $\rho \restriction_{\States} \in \obj_0$ or,
    \item $\rho$ has a good deviation.
    \end{enumerate}
    If every play $\rho$ in $\Plays(\sigma_0)$ is in $\obj_0$ then
    $\sigma_0$ is a solution for $NCRSP$. Let $\rho$ be a play in
    $\Plays(\sigma_0)$ such that it is not in $\obj_0$. By assumption,
    $\rho$ has a good deviation point i.e. there exists an Agent
    $i \neq 0$ and a strategy $\tau_i$ for the same agent such that:
    $i)$ $\rho\restriction_{\States} \not\in \obj_i$ and 
    $ii)$ after
    a finite prefix $h$ of $\rho$ for any tuple of strategies
    $(\sigma_j)_{j \in \Agt \setminus \{0,i\}}$ the following holds:
    \[
    [h] \cdot \out\left(\sigma_0^{h},
      \sigma_1^{h},\cdots,\tau_i^{h},\cdots, \sigma_n^{h}\right) \in
    \obj_i \enspace.
    \]
    Hence, $\rho$ is not the outcome of a 0-fixed NE and therefore
    $\sigma_0$ is a solution for the NCRSP.  
\end{proof}

\subsection{Correctness}

\begin{definition}
  $\cons$ wins if she has a strategy that ensures $\obj$ against any strategy of $\spoil$.
\end{definition}

\begin{proposition}
  If $\cons$ wins then there exists a solution for the NCRSP.
\end{proposition}

\begin{proof}
  Let $\sigma_E$ be a winning strategy for $\cons$ in $\HHH$, and let
  $\sigma_0$ be the strategy for Agent~$0$ in $\GGG$ obtained by the
  construction in Sec.~\ref{sec:GtoH} Equation~\eqref{eq:sigma0}, that
  is, for every history $h$ in $\GGG$,
  $\sigma_0(h) = \sigma_E(\GtoH_{\sigma_E}(h))[0]$. We show that
  $\sigma_0$ is solution to the NCRSP.

  Let $\rho$ be an arbitrary run in $\GGG$ compatible with $\sigma_0$.

  According to Lemma~\ref{lemma:devPoint} it is sufficient to show that
  % We need to show that
  % \[
  %   \rho \not\in \obj_0 \implies \rho \text{ is not the outcome of a
  %     0-NE.}
  % \]
  % Or equivalently, one of the two following equation holds.
  $
    \rho \text{ is in } \obj_0 \text{ or } \rho \text{ has a good deviation point.}
  $
  Consider the run $r = \GtoH^\bullet_{\sigma_E}(\rho)$ in $\HHH$.  As
  a consequence of Claim~\ref{claim:gtoh-play}, we have that $r$ is
  compatible with $\sigma_E$. Since $\sigma_E$ is winning, we also
  have $r \in \obj$, i.e.,
  \[
    r \in (S_0 \cup S_D) \cap S_W = (S_0 \cap S_W) \cup (S_D \cap S_W)
    \enspace.
  \]

  As a first case, assume that $ r \in S_0 \cap S_W $ implying
  $ \pi_1(r \restriction_\astate) \in \obj_0 $. By
  Claim~\ref{claim:GtoH-proj-runs} we can write
  $\pi_1(r \restriction_\astate) = \rho\restriction_\States$, and thus
  $\rho\restriction_\States \in \obj_0 $.

  As a second case, assume $ r \in S_D \cap S_W $.  It
  implies that there exists a state $q_a$ in $\astate$ along $r$ such
  that $ q_a = (s,W,D) $ and there exists an agent $i$ in $D$ such
  that $i$ in $\lim \pi_3(r\restriction_\astate)$ and
  $ \pi_1(r \restriction_\astate) \not\in \obj_i $.

  We argue that Agent $i$ has a profitable deviation from a prefix of
  $\rho$ entailing that $\rho$ contains a good deviation point.
  
  Assume w.l.o.g.~that $q_a$ is the first state along $r$ for which
  there exists an Agent $i$ in $D$ such that $i$ in
  $\lim \pi_3(r\restriction_\astate)$ and
  $ \pi_1(r \restriction_\astate) \not\in \obj_i $.
  The run $r$ is of the form:
  \begin{align}
    \label{eq:decomp}
    r =\tilde{h}\bar{a}p_b\bar{b}p_c\bar{c}p_d\bar{d}q_a\tilde{t}
  \end{align}
  where $\tilde{h}$ is a finite prefix of $r$ ending in a state in
  $\astate$, and $\tilde{t}$ is an infinite suffix.  Remember also
  that $r = \GtoH^\bullet_{\sigma_E}(\rho)$, hence there exists a
  history $h$ which is a prefix of $\rho$ such that
  $\tilde{h} =\GtoH(h)$.  We claim that $h$ is a good deviation point
  (c.f. Definition~\ref{def:devpoint}) for Agent $i$. Indeed, we use
  $\tau_i$ the a strategy defined only after $h$ has occurred as
  follows: $ \tau_i(h) = \bar{c}[i] $ where $\bar{c}$ is the action
  available for agent $i$ in Equation~\eqref{eq:decomp}, and for any
  history $hh'$ in $\GGG$:
  $ \tau_i(hh') = \sigma_E(\GtoH(hh'))[i] \enspace.  $ (Observe that
  by construction $\GtoH(hh')$ always ends in a state in $\astate$,
  controlled by $\cons$.)

  We define the set $T$ as the set of all the plays in $\GGG$ that
  start with $h$ and are compatible with $\tau_i$.  Let $\rho'$ be a
  play in $T$, and let $r' = \GtoH(\rho')$ be a play in $\HHH$. The
  play $r'$ enjoys two properties, first
  $i\in\lim\pi_2(r'\restriction_\astate)$ and second it is compatible
  with $\sigma_E$. Hence $ \pi_1(r'\restriction_\astate) \in \obj_i $.
  This shows that $\rho$ has a good deviation point after history
  $h$. % Therefore, $\rho$ is not the outcome of a 0-NE.
  By Lemma~\ref{lemma:devPoint} we conclude
  that $\sigma_0$ is solution to the problem NCRSP.
\end{proof}

%%% Local Variables:
%%% mode: latex
%%% TeX-master: "main.tex"
%%% End:

\subsection{Completeness}
\begin{proposition}
  \label{prop:compl}
  There exists a solution for the NCRSP then $\cons$ wins.
\end{proposition}

We first introduce some technical tools.
\begin{align*}
  \dev : \hist(\sigma_0) \times \Agt&\to \Act\cup\{-\}\\
  (h,i) &\mapsto 
      \begin{cases}
        a \text{ if $h$ is a good deviation point for Agent $i$ using
          action $a$,}\\
        - \text{ if not.}
      \end{cases}
\end{align*}

\begin{align*}
  \pref : \hist \times \Agt &\to \hist \cup \{\bot\}\\
  (h,i) &\mapsto
      \begin{cases}
        h' \text{ where $h'$ is the shortest prefix of $h$ s.t. }
        \dev(h',i) \in \Act\\
        \bot \text{ if no such a prefix exists}
      \end{cases}
\end{align*}

\begin{claim}
  \label{claim:compl}
  Let $h$ be a history and $i$ an agent s.t. $\dev(h,i) \in \Act$, then
  there exists a winning strategy $\tau_i$ from $\pref(h,i)$ for agent $i$.
\end{claim}

Indeed, assuming that $\dev(h,i) \in \Act$ and that there is no
winning strategy from $\pref(h,i)$, would entail that $\pref(h,i)$ is
not a good deviation point.

\begin{proof}[Proof of Proposition~\ref{prop:compl}]
  Let $\sigma_0$ be a solution for the NCRSP.
  Given a history $\tilde{h}$ in $\HHH$ s.t. $\last(\tilde{h})$ is in
  $\astate$, we let $h = \HtoG(\tilde{h})$. We construct a strategy
  $\sigma_E$ for $\cons$ as follows:
  $\sigma_E(\tilde{h}) = \bar{a}$
  such that $\bar{a}[0] = \sigma_0(h)$
  and for every $i$ in $W,$
  $\bar{a}[i] = \tau_i(h)$
  where $\tau_i$ is the strategy described by
  Claim~\ref{claim:compl}. Notice that this strategy is only defined for
  histories $h$ that satisfy $\pref(h,i) \neq \bot$. This is ensured
  because $i$ is in $W$, meaning that there exists a prefix $h'$ of $h$
  such that $h'$ is a good deviation point for Agent $i$.
  
  We also need to define $\sigma_E$ for histories ending in a
  $\cstate$. Consider any history of the form
  $\tilde{h}\bar{a}q_{b}\bar{b}q_c$, the strategy $\sigma_E$ is
  defined as follows:
  $ \sigma_E(\tilde{h}\bar{a}q_{b}\bar{b}q_c) = \bar{c} $ such that
  for every $i$ not in $W \cup D$, $ \bar{c}[i] = \dev(h,i) $ Let us
  show that $\sigma_E$ is winning for $\cons$. Let $r$ be a run
  compatible with $\sigma_E$. We must show that
  $r \in (S_0 \cup S_D) \cap S_W$.
  Denote $\rho = \HtoG^\bullet(r)$.
  By Claim~\ref{claim:HtoG} we have 
  $\pi_1(r \restriction_\astate) = \rho\restriction_\States$.
  
  If $\rho\restriction_\States \in \obj_0$ then $r \in S_0$.
  If $\rho\restriction_\States \not\in \obj_0$, since
  $\sigma_0$ is a solution, it follows that along $\rho$ some
  player has a good deviation point and is loosing.
  This entails that at some point $i$ will be in $D$ along $r$ i.e.
  $\dev(\tilde{h},i) \in \Act$
  for some $\tilde{h}$ a prefix of $r$. Thus
  $r \in S_D$.

  It remains to show that $r\in S_W$ this follows from the facts that
  1) any player in $W$ is due to the mapping $\dev$ that correctly
  guesses the correct deviations and 2) Claim~\ref{claim:compl}.
\end{proof}

%%% Local Variables:
%%% mode: latex
%%% TeX-master: "main.tex"
%%% End:

\section{Computational Complexity}
\label{sec:comp}

In this section, we take advantage of the construction presented in
the previous section to give complexity bounds for a variety of winning
conditions.  In fact, we can adapt the technique used in~\cite{ICALP16} in order to establish the upper bound complexity for
NCRSP.  

In the case of Reachability, Safety, B\"{u}chi and coB\"{u}chi
conditions, we reduce the game $\HHH$ to a finite duration game.  We
actually tranform the winning condition into a finite horizon
condition in a finite duration game.
In order to obtain this finite duration game, we simply rewrite the
winning objective of $\HHH$ and obtain a new game $\HHH'$ with Parity
objective. The plays in the finite duration game $\HHH^f$ are obtained
by stopping the plays in the game $\HHH'$ after the first loop.

In the remainder of this section, for technical convenience, we assume
that the histories in $\HHH$ are defined over the set $Q^*$ and that
the plays are defined over $Q^\omega$. This does not affect the
validity of the results since $\Tab'$ is deterministic and the actions
are encoded in the states.

When the game $\HHH'$ is equipped with a
parity condition, we construct a finite duration game
that stops after the first loop. In particular, if
$\pr : \States' \rightarrow \mathbb{N}$ is the priority function in
$\HHH'$.  Then, the finite duration game $\HHH^f$ is defined over the
same game structure as $\HHH'$, but each play stops after the first
loop. We will consider such a play winning if the least parity in the
loop is even i.e a play $r = xy_1y_2y_3\ldots y_l y_1$ where
$x\in q_0Q^*$ and $y_1,y_2,\ldots,y_l \in Q$ is winning for $\cons$ if
$\min\{\pr(y_k) \mid 0 \leq k\leq l\}$ is even.

The following lemma establishes the relation between $\HHH'$ and
$\HHH^f$. It is actually a consequence of a result that appeared
in~\cite{AminofR17}.

\begin{lemma}\label{lemma:Finite}
  $\cons$ has a winning strategy in the game $\HHH'$ with the
  parity condition $\pr$ if and only if she has a winning strategy in
  the game $\HHH^f$.
\end{lemma}

The following lemma establishes the fact that inside a cycle in the
game $\HHH$, the values of the sets $W$ and $D$ do not change.

\begin{lemma}
  \label{lemma:WD_constant}
  Let $r$ be a play in $\HHH$, and consider a loop along $r$. Let also
  $q$ and $q'$ be two states on this loop. We have
  $\pi_2(q) = \pi_2(q')$ and $\pi_3(q) = \pi_3(q')$.
\end{lemma}

\begin{proof}
  Let $r = x q y q z$ be an infinite play in $\HHH$. From the
  definition of the transition relation in $\HHH$,
  $r' = x (q y)^\omega$ is also a valid play in $\HHH$.  Then, assume
  towards contradiction that there are two states in $qy$ having
  different values on states $W$ or $D$. It means that there is at
  least one player $i$ that is added or removed to/from $W$ or
  $D$. Therefore, along $r'$ we would have an infinite number of
  additions or removals to/from $W$ or $D$.  But, according to the
  transition relation, this is not possible because once a player is
  removed from $W$, it is added to the set $D$ and never added to $W$
  again along $r'$. Also, once a player added in the set $D$, he is
  never removed.  Therefore, along each path, along each loop, the
  values of $W$ and $D$ do not change.
\end{proof}

In order to check the condition of  reachability, we keep along
plays in the game $\HHH$ a set $P$ of players in the environment that
already have visited their target states.  Then, the resulting game
$\HHH'$ has states in $Q \times 2^\Agt$ where $Q$ is the
set of states in the game $\HHH$.
The set $P$ is initially equal
to the set of players for which the initial state is in their target
set.  Let $R_i \subseteq \States$ be the target set of Player $i$.
Then, $P_0 = \{ i \mid s_0 \in R_i \}$ and the initial state in the
resulting game $\HHH'$ is $(s_0, \emptyset, \emptyset, P_0)$.
 
The set $P$ is updated as follows:
\begin{itemize}
\item if $(q,q') \in \Tab'$ in $\HHH$ and
  $q' = (s,W,D) \in \States \times 2^\Agt \times 2^\Agt$, then
  $((q,P),(q',P \cup \{ i \mid s \in R_i \} ))$ is the corresponding
  transition in $\HHH'$.
\item if $(q,q') \in \Tab'$ in $\HHH$ and
  $q' \not \in \States \times 2^\Agt \times 2^\Agt$, then
  $((q,P),(q',P))$ is the corresponding transition in $\HHH'$.
\end{itemize}

Note that the set $P$ also eventually stabilizes since it only
increases and there is a finite number of players in $\GGG$. Let
$\lim \pi_4(r \restriction_{\astate})$ be the limit along the play
$r$.

The objective of $\cons$ in the game $\HHH$ is
written in $\HHH'$ as the Büchi condition
\textsc{Büchi($F^R$)} where:
\[
F^R = \{ (s,W,D,P) \mid ( 0 \in P \text{ or } D \setminus P \neq
\emptyset ) \text{ and } ( W \subseteq P ) \}
\enspace.
\]

Now, the B\"{u}chi objective \textsc{Büchi($F^R$)} can be expressed as
the parity objective \textsc{Parity($\pr$)} with $\pr(v) = 0 $ if
$ v = (s,W,D,P) \in F^R$ and $\pr(v) = 1$ otherwise.

We now define the finite duration game $\HHH^f$ over the same game
arena as $\HHH$, but each play stops when the first state in
$\States \times 2^\Agt \times 2^\Agt \times 2^\Agt$ is repeated.
Then, each play is of the form $r = xy_1y_2y_3\ldots y_l y_1$ where
$x\in q_0(Q')^*$ and $y_1,y_2,\ldots,y_l \in Q'$ with
$Q' = \States \times 2^\Agt \times 2^\Agt \times 2^\Agt$.  $\cons$
wins in the game $\HHH^f$ if $y_1= (s,W,D,P)$ is such that
$( 0 \in P \text{ or } D \setminus P \neq \emptyset ) \text{ and } ( W
\subseteq P )$.  Equivalently, thanks to Lemma \ref{lemma:WD_constant}
and because the value of $P$ does not change for the same argument
that show $W$ and $D$ eventually stabilize. Finally $\cons$ wins if
$\min\{ \pr(y_k) \mid 0 \leq j < l \}$ is even.

\begin{lemma}\label{lema:ReachSize}
  All plays in the game $\HHH^f$ have polynomial length in the size of
  the initial game.
\end{lemma}

\begin{proof}
  Since $D$ and $P$ are monotone, there are at most $|\Agt|+1$
  different values that they can take on a path of $\HHH$.  Also, in
  the set $W$ we can have at most one addition and one removal for
  each player $i \in \Agt$ and hence $2|\Agt|+1$ different values for
  $W$. Therefore, along a play $\pi$ there are at most
  $r = 1+(2|\Agt|+1) \cdot (|\Agt|+1)^2 \cdot |\States| $ different
  states in $\States \times 2^\Agt \times 2^\Agt \times 2^\Agt$.
  Then, between two states in
  $\States \times 2^\Agt \times 2^\Agt \times 2^\Agt$, there are three
  intermediate states.  Therefore, since all the plays in $\HHH^f$
  stop after the first cycle, the length of each play is of at most
  $4r+1$ states since there is only one state that appears
  twice. Therefore, all plays in $\HHH^f$ have polynomial length in
  $\Agt$ and $\States$ of the initial play $\GGG$.
\end{proof}

\begin{proposition}
  Deciding if there is a solution for the non-cooperative rational
  synthesis in concurrent games with Reachability objectives is in
  \textsc{PSpace}.
\end{proposition}

\begin{proof}
  Using Lemmas \ref{lemma:Finite} and \ref{lema:ReachSize}, solving
  the non-cooperative rational synthesis problem, reduces to solving
  the finite duration game $\HHH^f$ which has polynomial length
  plays. This can be done in \textsc{PSpace} using an alternating
  Turing machine runing in \textsc{PTime}.
\end{proof}

In the case of Safety, B\"{u}chi and coB\"{u}chi conditions, we
essentially use similar constructions; c.f.~Appendix for details on the
constructions. Roughly speaking, In the case of safety it sufficient
to ``dualize'' the winning condition. In the cases of B\"{u}chi and
coB\"{u}chi objectives, the idea is to transform the game $\HHH$ by
possibly adding some counters such that $\cons$'s objective can be
written as a parity objective.  Note that these constructions are
similar to the ones in \cite{ICALP16}.  In the case of Muller
conditions, we have to use Least Appearance Record (LAR) construction
to get the parity game $\HHH'$ and then the finite duration game would
have plays with exponential length in the size of the initial
game. This approach would give \textsc{ExpSpace}
complexity. Fortunately, the parity condition in the game $\HHH'$ that
we obtain after applying the LAR construction has an exponential
number of states but only a polynomial number of priorities. Then, by
using the result from~\cite{JPZ06,Schewe07}, we obtain
\textsc{ExpTime} complexity.

\begin{theorem}\label{thm:complex}
  Deciding if there is a solution for the non-cooperative rational
  synthesis problem in concurrent games is in \textsc{PSpace} for
  Safety, Reachability, B\"{u}chi and co-B\"{u}chi objectives and
  \textsc{ExpTime} for Muller objectives.
\end{theorem}

% \subsection{Fixed Number of Agents}

In the case of a fixed number of agents, the game $\HHH$ that we build
has polynomial size in the size of the initial game $\GGG$ (when
considering that the transitions are given explicitly in the table
$\Tab$ since we build nodes in $\HHH$ for each possible action
profile).  This lowers the complexities that we obtain for the
rational synthesis problem.  The theorem below holds because the game
$\HHH$ has polynomial size and $\cons$'s objective is
fixed.

\begin{theorem}
  Deciding if there is a solution for the non-cooperative rational
  synthesis in concurrent games with a fixed number of agents and
  Safety, Reachability, B\"{u}chi or co-B\"{u}chi objectives is in
  \textsc{PTime}.
\end{theorem}

% The above argument does not work in the case of Muller conditions
% However, as it is also done in \cite{ICALP16}, it is
% easily rewritten as a two-player Muller 
% game of polynomial size. Moreover, two-player Muller games can be
% solved using polynomial space. This
% leads the following complexity result.

\begin{theorem}
  Deciding if there is a solution for the non-cooperative rational
  synthesis in concurrent games with a fixed number of agents and
  Muller objectives is in \textsc{PSpace}.
\end{theorem}

%%% Local Variables:
%%% mode: latex
%%% TeX-master: "main.tex"
%%% End:

\section{Conclusions}
\label{sec:concl}
In some circumstances, the Non Cooperative Rational Synthesis Problem (NCRSP) introduced in~\cite{KuPV14} and defined here as Problem~\ref{pb:NCRSP} might arguably accept undesired solutions.
It asks whether there is a strategy $\sigma_0$ for Agent $0$ such that \emph{for every} $0$-NE, if $\bar{\sigma} = \langle\sigma_0, \bar{\sigma}_{-0} \rangle $, then $\pay_0(\bar{\sigma}) = 1$.
One sees the objective of Agent~$0$ as a critical property satisfied by all rational evolutions of the system.
A possibly unwanted consequence is that a strategy $\sigma_0$ which does not allow any rational evolution of the system, thus forcing anarchy, would be a solution.
The original definition of NCRSP can be strengthened so as to ask for a strategy $\sigma_0$ for Agent~$0$ such that \emph{there is} at least one $0$-NE.
Another amendment can also restrict the class of game structures. For instance, one can consider pseudo turn-based games, where $0$-NE are certain to exist. It suffices to add in Definition~\ref{def:cgs} the constraint that in every state, only Agent~$0$ and at most one other agent have non-vacuous choices. The games are still concurrent. Agent~$0$ can still effectively control every state, but once her strategy $\sigma_0$ is fixed, the sub-game induced by $\sigma_0$ has all the characteristics of a turn-based game, where there is always a $0$-NE.

% \section{Solution for Problem~\ref{pb:QNCRSP}}
% \subsection{Construction}
% \input{quantitative.tex}

%%
%% Bibliography
%%

%% Please use bibtex, 

\bibliography{ref.bib}

\newpage
\appendix
% \section{Proofs of Section~\ref{}}

%\section{Proof of Lemma~\ref{lemma:WD_constant}}

%\section{Proof of Lemma~\ref{lema:ReachSize}}

\section{Safety Objectives}

In the case of Safety objectives, we use the set $P$ with the
following semantics: all players that are in the set $s$ already lost
by passing through an unsafe state.

Initially, the set $P$ equals to the set of players for which the
initial state is unsafe. Let $S_i$ be the set of safe states of Player
$i$. Then, $P_0 = \{ i \mid s_0 \not \in S_i \}$.

The set $P$ is updated as follows:
\begin{itemize}
\item if $(q,q') \in \Tab'$ in $\HHH$, and
  $q' = (s, W, D) \in \States \times 2^\Agt \times 2^\Agt$, then
  $( (q,P), (q', P \cup \{ i \mid s \not \in S_i \}) )$ is the
  corresponding transition in $\HHH'$

\item if $(q,q') \in \Tab'$ in $\HHH$, and
  $q' \not \in \States \times 2^\Agt \times 2^\Agt$, then
  $( (q,P), (q', P) )$ is the corresponding transition in $\HHH'$

\end{itemize}

Note that in this case the set $P$ is also increasing and eventually
stabilizes to a limit $\lim \pi_4(r \restriction_{\astate})$ along the
play $r$.

Then, using the fact that the sets $W$, $D$ and $P$ eventually
stabilize, the objective of $\cons$ can be rewritten as the B\"{u}chi
objective \textsc{B\"{u}chi($F^S$)} where
\[
  F^S = \{ (s,W,D,P) \mid ( 0 \not \in P \text{ or } D \cap P \neq
  \emptyset ) \text{ and } ( W \cap P = \emptyset ) \}
\]

Using a similar proof as in the case of Reachability objectives, we
can prove that we can reduce to solving a finite duration game having
plays of polynomial length.  Therefore, the following theorem holds.

\begin{proposition}
  Answering the rational synthesis problem in concurrent games with
  Safety objectives is in \textsc{PSpace}.
\end{proposition}

\section{Büchi Objectives}

In the case of B\"{u}chi objectives, $\cons$'s objective is
\begin{align*}
  \obj = \big\{ r \in Q^\omega \mid & \ \big( (
                                      \pi_1(r \restriction_{\astate}) \models \Box \Diamond F_0 )  \text{ or } ( \exists
                                      i \in \lim\pi_3(r \restriction_{\astate}) \text{ s.t. } \pi_1(r
                                      \restriction_{\astate}) \models \Diamond \Box \neg F_i ) \big ) \\
                                    & \text{ and } \big ( \forall i \in \lim\pi_2(r \restriction_{\astate})
                                      \implies\pi_1(r \restriction_{\astate}) \models \Box \Diamond F_i  \big ) 
                                      \big\} \enspace.
\end{align*}

In order to reduce to the Parity objectives, we first make some small
changes on $\cons$'s objective as follows.  We exploit the fact that
the sets $D$ and $W$ eventually stabilize and rewrite the formulas
\begin{align*}
  &\phi_D \equiv \exists
    i \in \lim\pi_3(r \restriction_{\astate}) \text{ s.t. } \pi_1(r
    \restriction_{\astate}) \not \models \Box \Diamond  F_i  \text{ \hspace{10pt} and } \\
  & \phi_W \equiv \forall i \in \lim\pi_2(r \restriction_{\astate})
    \implies\pi_1(r \restriction_{\astate}) \models \Box \Diamond F_i
\end{align*}

First, the negation of $\phi_D$ says that for all players in
$ \lim\pi_3(r \restriction_{\astate})$, holds
$\pi_1(r \restriction_{\astate}) \models \Box \Diamond F_i$.  Since
the set $D$ stabilize, and the formulas to be verified inside $\phi_D$
is a tale objectives, instead of
$\lim\pi_3(r \restriction_{\astate})$, we can consider the current
value of the set $D$ and use a counter $c_D$ to wait for each player
$i \in D$ (on turns) a state
$q' = (s',D',W') \in \States \times 2^\Agt \times 2^\Agt$ s.t.
$s' \in F_i$.  Then, the formula $\neg \phi_D$ is satisfied if either
$D = \emptyset$ or we visit infinitely often a state $(q,c_d)$ with
$q = (s,D,W)$ and the counter $c_D$ takes the smallest value in $D$
and $s \in F_{c_D}$.

For the formula $\phi_W$ we proceed in the same way. We consider the
value of the set $W$ along executions and use a counter $c_W$ to
``check'' the appearance of a state $q = (s,D,W)$ such that $s\in F_i$
for each player $i \in W$.

Formally, the obtained game $\tilde{\HHH}$ is as follows: the set of
states $\tilde{Q}$ consists of tuples $(q,c_D, c_W)$ where $q$ is a
state in $\HHH$; $((s_0,\emptyset,\emptyset), -1, -1)$ is the initial
state; and transition between states is as follows:
\begin{itemize}
\item $(q,c_D, c_W) \rightarrow (q', c_D, c_W)$ iff $(q,q')$ is a
  transition in $\HHH$ and $q \in \States \times 2^\Agt \times 2^\Agt$

\item $(q,c_D, c_W) \rightarrow (q', c'_D, c'_W)$ $(q,q')$ is a
  transition in $\HHH$ and
  $q' = (s',D',W') \in \States \times 2^\Agt \times 2^\Agt$ and

  $ c'_D =
  \begin{cases}
    -1 & \text{ if }  D' = \emptyset \\
    \min\{(c_D+\ell)  \pmod{n} \in D' \mid \ell > 0\} & \text{ if }   D' \neq \emptyset \text{ and } (s' \in F_{c_D} \text{ or } c_D = -1 ) \\
    c_D & \text{ otherwise }
  \end{cases}
  $
 
  and

  $ c'_W =
  \begin{cases}
    -1 & \text{ if } W' = \emptyset \\
    \min\{(c_W+\ell) \pmod{n} \in W' \mid \ell > 0\} & \text{if } W'
    \neq \emptyset \text{ and } (v \in F_{c_W}
    \text{ or } c_W \not \in W' \text{ or } c_W = -1) \\
 
    c_W & \text{\hspace{-70pt} otherwise }
  \end{cases}
  $
 
\end{itemize}

Also, for a play $r \in \tilde{Q}^\omega$ we have that
$\pi_1(r \restriction_{\astate}) \models \Box \Diamond F_0$ if the
corresponding play $\tilde{r}$ for $r$ in $\tilde{\HHH}$ satisfies
$\tilde{r} \models \Box \Diamond T_0$ where
$T_0 = \{ q = (s,W,D,c_D, c_W) \in \States \times 2^\Agt \times 2^\Agt
\times (\Agt \cup \{-1\}) \times (\Agt \cup \{-1\}) \mid s \in F_0 \}$
 
Let
$C_1 = \States \times 2^\Agt \times 2^\Agt \times (\Agt \cup \{-1\})
\times (\Agt \cup \{-1\})$ be the set of states $q = (s,W,D,c_D, c_w)$
of $\cons$ and $r\restriction{_{C_1}}$ be the restriction of
$\tilde{r} \in \Plays(\tilde{\HHH})$ on $C_1$.  Then, the objective
$\obj$ can be equivalently written in the game $\tilde{\HHH}$ as
\[
  \tilde{\obj} = \set{ \tilde{r} \in \tilde{Q}^\omega \mid
    \tilde{r}\restriction{_{C_1}} \models (\Box \Diamond T_0 \vee
    \Diamond \Box \neg T_d) \wedge \Box \Diamond T_w }
\]
 
where
$ T_d = \{ (s,W,D,c_D, c_w) \mid D = \emptyset \vee ( s \in F_{c_D}
\wedge c_D = \min \{ i \in D \} ) \} $ and
$ T_w = \{ (s,W,D,c_D,c_W) \mid W = \emptyset \vee ( s \in F_{c_W}
\wedge c_W = \min \{ i \in W \} ) \} $.
 
To continue, the formula
$(\Box \Diamond T_0 \vee \Diamond \Box \neg T_d) \wedge \Box \Diamond
T_w$ is equivalent to
$(\Box \Diamond T_0 \wedge \Box \Diamond T_w) \vee ( \Diamond \Box
\neg T_d \wedge \Box \Diamond T_w)$ and we also use a counter (bit)
$b \in \{0,1\}$ to verify $\Box \Diamond T_0 \wedge \Box \Diamond T_w$
and therefore the set of states in the new game $\HHH'$ denoted $Q'$
consists of tuples of the form $(q,c_D, c_W, b)$.  Initially, $b = 0$
and the transition relation is as follows:
$(q,c_D, c_W, b) \rightarrow (q', c_D, c_W, b')$ iff
$(q,c_D, c_W) \rightarrow (q', c_D, c_W)$ is a transition in
$\tilde{\HHH}$ and
\[
  b' =
  \begin{cases}
    1 & \text{ if } b = 0 \text{ and } (q,c_D, c_W) \in T_0 \\
    0 & \text{ if } b = 1 \text{ and } (q,c_D, c_W) \in T_w \\
    b & \text{ otherwise }
  \end{cases}
\]
 
Then, considering $C'_1 = C_1 \times \{0,1\}$, the winning objective
is
\[
  \obj' = \{ r' \in Q'^\omega \mid r'\restriction{_{C'_1}} \models
  \Box \Diamond T'_0 \vee ( \Diamond \Box \neg T'_d \wedge \Box
  \Diamond T'_w)
\]
where $T'_0 = \{ (q,c_D, c_W, 0) \mid (q,c_D, c_W) \in T_0 \}$,
$T'_d = T_d \times \{0,1\}$ and $T'_w = T_w \times \{0,1\}$.
 
And finally, we have that a a play $r'$ satisfies
$r'\restriction{_{C'_1}} \models \Box \Diamond T'_0 \vee ( \Diamond
\Box \neg T'_d \wedge \Box \Diamond T'_w)$ iff the Parity condition
$Parity(\pr)$ is satisfied by $r'$ where the priority function $\pr$
is defined as follows: For
$q' = (s,D,W,c_D,c_W,b) \in \States \times 2^\Agt \times 2^\Agt \times
(\Agt \cup \{-1\}) \times (\Agt \cup \{-1\}) \times \{0,1\}$,
\[
  \pr(q') =
  \begin{cases}
    0 & \text{ if } q' \in T'_0 \\
    1 & \text{ if } q' \not \in T'_0 \wedge  q' \in T'_d \\
    2 & \text{ if } q' \not \in T'_0 \wedge q' \not \in T'_d \wedge q' \in T'_w \\
    3 & \text{ if } q' \not \in T'_0 \wedge q' \not \in T'_d \wedge q'
    \not \in T'_w
  \end{cases}
\]
 
For
$q' \not \in \States \times 2^\Agt \times 2^\Agt \times (\Agt \cup
\{-1\}) \times (\Agt \cup \{-1\}) \times \{0,1\}$, $\pr(q') = 4$.
 
Since each play in the parity game $\HHH'$ has polynomial number of
different states, we can use Lemmas \ref{lemma:Finite} and obtain a
finite duration game whose plays have polynomial length. This gives
the following result:
 
 \begin{proposition}
   Deciding if there is a solution for the non-cooperative rational
   synthesis in concurrent games with B\"{u}chi objectives is in
   \textsc{PSpace}.
 \end{proposition}

 \section{co-B\"{u}chi Objectives}
 
 For co-B\"{u}chi objectives, the winning condition for $\cons$ in the
 game $\HHH$ is
 
 \begin{align*}
   \obj = \big\{ r \in Q^\omega \mid & \ \big( (
                                       \pi_1(r \restriction_{\astate}) \models  \Diamond \Box \neg F_0 )  \text{ or } ( \exists
                                       i \in \lim\pi_3(r \restriction_{\astate}) \text{ s.t. } \pi_1(r
                                       \restriction_{\astate}) \models  \Box \Diamond F_i ) \big ) \\
                                     & \text{ and } \big ( \forall i \in \lim\pi_2(r \restriction_{\astate})
                                       \implies\pi_1(r \restriction_{\astate}) \models  \Diamond \Box \neg F_i  \big ) 
                                       \big\} \enspace.
 \end{align*}  
 
 We use again the fact that the sets $D$ and $W$ stabilize along a
 play $r$ and the fact that co-B\"{u}chi objectives are tail
 objectives.  Let $C_1 = \States \times 2^\Agt \times 2^\Agt$ be the
 set of states $(s, D, W)$ of $\cons$.  Then,
 $ \exists i \in \lim\pi_3(r \restriction_{\astate}) \text{ s.t. }
 \pi_1(r \restriction_{\astate}) \models \Box \Diamond F_i $ is
 equivalent to $r\restriction_{C_1} \models \Box \Diamond T_d $ where
 $T_d = \{ q = (s,D,W) \mid s \in \bigcup_{i \in D} F_i \}$.
  
 Further,
 $\forall i \in \lim\pi_2(r \restriction_{\astate}) \implies\pi_1(r
 \restriction_{\astate}) \models \Diamond \Box \neg F_i$ is equivalent
 to $r\restriction_{C_1} \models \Diamond \Box \neg T_w $ where
 $T_w = \{ q = (s,D,W) \mid s \in \bigcup_{i \in W} F_i \}$.

 Therefore, the winning condition for $\cons$ in the game $\HHH$ is
 equivalently written as

\[
  \obj= \{ r \in Q^\omega \mid r\restriction_{C_1} \models
  (\Diamond \Box \neg T_0 \vee \Box \Diamond T_d) \wedge \Diamond \Box
  \neg T_w \}\enspace.
\]

This can be written as the Parity condition $Parity(\pr)$ where the
priority function $\pr $ is defined as follows: For
$q \in \States \times 2^\Agt \times 2^\Agt$,
\[
  \pr(q) =
  \begin{cases}
    1 &\text{ if } q \in T_w \\
    2 &\text{ if } q \not \in T_w \text{ and } q \in T_0 \cap T_d  \\
    3 &\text{ if } q \not \in T_w \cup T_d \text{ and } q \in T_0\\
    4 &\text{ if } q \not \in T_w \cup T_0
  \end{cases}
\]
For $q \not \in \States \times 2^\Agt \times 2^\Agt$, $\pr(q) = 6$.

Now, applying Lemma \ref{lemma:Finite} on the game $\HHH$ with parity
objective $Parity(\pr)$, and since each play in $\HHH$ has a
polynomial number of distinct states, we get the following complexity
result.

\begin{proposition}
  Deciding if there is a solution for the non-cooperative rational
  synthesis in concurrent games with co-B\"{u}chi objectives is in
  \textsc{PSpace}.
\end{proposition}

\section{Muller Objectives}

Let $\mu_i$ be the Muller objective of Player $i$. Then, the $\cons$'s
objective in the game $\HHH$ is

\begin{align*}
  \obj = \big\{ r \in Q^\omega \mid & \ \big( 
                                      \pi_1(r \restriction_{\astate}) \in \textsf{Muller} (\mu_0)   \text{ or }  \exists  i \in \lim\pi_3(r \restriction_{\astate}) \text{ s.t. } \pi_1(r
                                      \restriction_{\astate}) \not \in \textsf{Muller}(\mu_i)  \big ) \\
                                    & \text{ and } \big ( \forall i \in \lim\pi_2(r \restriction_{\astate})
                                      \implies\pi_1(r \restriction_{\astate}) \in \textsf{Muller}(\mu_i)  \big ) 
                                      \big\} \enspace.
\end{align*}  

Contrary to the previous cases, in the case of Muller objectives, we
cannot directly reduce to a finite duration game with plays having
polynomial length.  Instead, as also proceeded in \cite{ICALP16}, we
use \textit{Least Appearance Record}(LAR) construction to reduce the
objective $\obj$ to a parity objective with a polynomial number of
priorities.  That is, each state in the obtained game $\HHH'$ is of
form $(q, (m,h))$ where $q$ is a state in $\HHH$, $m \in P(\States)$
is a permutation of states in $\States$ and
$h \in \{ 0, 1,...,|\States|-1 \}$ is the position in $m$ of the last
state $s$ that appeared in $q$.

The transition between states is defined by:

\begin{itemize}
\item $(q, (m,h)) \longrightarrow (q', (m,h))$ if $q \rightarrow q'$
  in $\HHH$ and $q' \not \in \States \times 2^\Agt \times 2^\Agt$

\item $(q, (m,h)) \longrightarrow (q', (m',h'))$ if $q \rightarrow q'$
  in $\HHH$ and $q' \in \States \times 2^\Agt \times 2^\Agt$ where,
  assuming $q' = (s, D, W)$ and $m = x_1 s x_2$ for some
  $x_1, x_2 \in \States^*$, $(m',h') = (x_1 x_2 s, |x_1|)$
\end{itemize}

Finally, the priority function $\pr$ over states in $\HHH'$ is defined
as:

\begin{itemize}
\item for $q = (s,D,W) \in \States \times 2^\Agt \times 2^\Agt$,
  \[
    \pr((s,D,W),(m,h)) =
    \begin{cases}
      2h & \hspace{-10pt}\text{ if } \forall i \in W  \{ m[l] \mid l \geq h \} \models \mu_i \text{ and } \\
      &
      \hspace{-15pt} (\{ m[l] \mid l \geq h \} \models \mu_0 \text{ or } \exists i \in D \text{ s.t. } \{ m[l] \mid l \geq h \} \models \neg \mu_i) \\
      2h+1 & \text{ otherwise }
    \end{cases}
  \]

\item for $q \not \in \States \times 2^\Agt \times 2^\Agt$,
  $\pr(q,(m,h)) = 2 |\States| + 2$.
\end{itemize}

Note that in this case, if we use the reduction to finite duration
game, we obtain exponential size plays.  Instead, we use the fact that
the game $\HHH'$ has exponential number of states in the size of the
original game $\GGG$, but it has a Parity objective with polynomial
number of priorities.  Then, the results in \cite{JPZ06,Schewe07}
prove the following theorem:

\begin{proposition}
  Deciding if there is a solution for the non-cooperative rational
  synthesis in concurrent games with Muller objectives is in
  \textsc{ExpTime}.
\end{proposition}

%%% Local Variables:
%%% mode: latex
%%% TeX-master: "main.tex"
%%% End:

\end{document}